\documentclass[a4paper,fleqn,10pt]{cas-dc}
\usepackage[numbers]{natbib}
\usepackage{amsmath,amssymb,amsthm}
\usepackage{color}
\usepackage{chngcntr}
\usepackage{indentfirst}
\usepackage{enumitem}
\usepackage{hyperref}
\usepackage{bbm}
\usepackage{siunitx}
\usepackage{xcolor}
\usepackage[linesnumbered,ruled,vlined]{algorithm2e}

\SetCommentSty{mycommfont}

\SetKwInput{KwInput}{Input}                
\SetKwInput{KwOutput}{Output}              

\usepackage{optidef}
\usepackage{cleveref}

\def\tsc#1{\csdef{#1}{\textsc{\lowercase{#1}}\xspace}}
\tsc{WGM}
\tsc{QE}
\tsc{EP}
\tsc{PMS}
\tsc{BEC}
\tsc{DE}

\renewcommand{\printorcid}{}  

\theoremstyle{case}

\newtheorem{remark}{Remark}
\newtheorem{prop}{Proposition}[section]
\newtheorem{theorem}{Theorem}[section]

\newtheorem{corollary}{Corollary}[theorem]

\newtheorem*{theorem*}{Theorem}
\newtheorem*{prop*}{Proposition}

\DeclareMathOperator*{\argmin}{arg\,min}

\usepackage{tikz}
\usetikzlibrary{matrix,positioning}
\tikzset{bullet/.style={circle,fill,inner sep=2pt}}
\usepackage{adjustbox}
\usepackage{pgfplots}

\begin{document}
\let\WriteBookmarks\relax
\def\floatpagepagefraction{1}
\def\textpagefraction{.001}
\shorttitle{DRPOs}
\shortauthors{Singh et~al.}

\title [mode = title]{Distributionally Robust Profit Opportunities}                      



\author[1]{Derek Singh}
\cormark[1]
\ead{singh644@umn.edu}
\cortext[cor1]{Corresponding author}

\author[1]{Shuzhong Zhang}
\ead{zhangs@umn.edu}

\address[1]{Department of Industrial and Systems Engineering, University of Minnesota, Minneapolis, MN 55455}

\ExplSyntaxOn
\keys_set:nn { stm / mktitle } { nologo }
\ExplSyntaxOff

\begin{abstract}
This paper expands the notion of robust profit opportunities in financial markets to incorporate distributional uncertainty using Wasserstein distance as the ambiguity measure. Financial markets with risky and risk-free assets are considered. The infinite dimensional primal problems are formulated, leading to their simpler finite dimensional dual problems. A principal motivating question is how does distributional uncertainty help or hurt the robustness of the profit opportunity. Towards answering this question, some theory is developed and computational experiments are conducted. Finally some open questions and suggestions for future research are discussed.
\end{abstract}

\begin{keywords}
robust profit opportunities \sep Sharpe ratio \sep distributionally robust optimization \sep Wasserstein distance \sep Lagrangian duality
\end{keywords}

\maketitle

\section{Introduction and Overview}
\subsection{Profit Opportunities in Financial Markets}
\indent Modern financial markets cover a wide array of asset classes including (but not limited to) stocks, bonds, loans, money market instruments, currencies and commodities, real estate, derivatives, and so on. The concept of a profit opportunity (through favorable purchase and sale of securities) is as old as financial markets themselves. Various trading and investment strategies have been developed, using advances in technology and quantitative methodologies, to identify and monetize such profit opportunities in modern financial markets. Risk adjusted return is one class of performance metrics used to evaluate the attractiveness of such opportunites. One well known example of this is the Sharpe Ratio which looks at the ratio of expected excess return to risk as measured by variance of the return. A modern revision of this uses a benchmark index to measure excess return and its variance \cite{sharpe1994sharpe}. In \cite{chen2011all} the authors show the utility of this metric and its linkage to other risk metrics such as the Sortino ratio, Omega ratio, CVaR ratio, and others under a Q-radial distributional assumption for returns. \par

The notion of a robust profit opportunity (RPO) for risky assets and its relation to the Sharpe Ratio were first introduced and discussed in \cite{pinar2005robust}.
The RPO can be seen as a relaxation of the notion of an arbitrage opportunity towards one of statistical arbitrage; a term referring to arbitrage that is statistically likely but not certain to occur. The parameter $\theta$, which measures the robustness of the profit opportunity, quantifies the number of standard deviations the asset returns could drop and yet the investment would still break even or generate some profit. \par

The purpose of this work is to extend the notion of an RPO to a setting that incorporates ambiguity about the underlying distribution of risky asset returns. This is done via the framework of Wasserstein discrepancy between distributions and the corresponding infinite dimensional Lagrangian duality results. The first steps are to define a notion of distributionally robust profit opportunities and formulate a primal problem that measures the effect of ambiguity in distribution, as measured by $\delta$, on the degree of robustness as measured by $\theta$. With that in hand, next steps are to formulate and solve the simpler finite dimensional dual problems to quantify the lower and upper bounds for robustness $\theta$ as a function of ambiguity $\delta$. 
An outline of this paper is as follows. Section 1 gives on overview of the financial concepts of profit opportunities and robustness as well as a literature review. Section 2 develops the main theoretical results to characterize robust profit opportunities for financial markets with risky and risk free assets. Section 3 conducts a case study of distributionally robust profit opportunities using a five year historical data set of month end closing prices for a basket of exchange traded funds (ETFs) spread across different sectors of the economy. Section 4 discusses conclusions and suggestions for further research. All detailed proofs are deferred to the Appendix. \par

\subsection{Literature Review}
In conducting the literature review for this research, not many references were found that have investigated the topic of statistical arbitrage under distributional uncertainty. From Section 1.1 above, one can see that considerable research has been done in academic (and industry) circles regarding the classical notions of statistical arbitrage in financial markets. Indeed, several academic papers and financial textbooks have been written that cover these topics from their origin in the 1980s until today. It was surprising to us, at least, to find only a few papers that address and/or extend the classical notions of statistical arbitrage under the presence of some form of distributional uncertainty. This subsection gives an overview of what we found in the academic literature. \par
One seminal paper of note by Ostrovskii \cite{ostrovski2013stability} introduced the notion of robust arbitrage under distributional uncertainty.
Ostrovskii used the total variation (TV) metric to characterize a radius $\delta_{TV}$ such that all probability measures $Q'$ within this distance from a weak arbitrage free reference measure $Q$ are also weak arbitrage free. The author remarks that $\delta_{TV}$ can be interpreted as the minimal probability of success that a zero cost initial portfolio $w \in \mathbb{R}^n$ achieves positive value $w \cdot S_1$ at time 1. The main result (and intermediate results) relating $\delta_{TV}$ to the minimal probability of success are established via proof by contradiction using tools from probability theory and real analysis. This work was extended in \cite{singh2020robust} to consider the Wasserstein metric and investigate a relaxed notion of classical arbitrage defined as statistical arbitrage. \par
A recent paper, \cite{zhao2017distributionally}, investigated the behavior of reward-risk ratios, in particular the Sortino-Satchel and Stable Tail Adjusted Return ratios (both modern variations of the Sharpe ratio), under distributional uncertainty in the Wasserstein framework. The authors provide tractable convex dual reformulations of these infinite dimensional primal problems using recent results from \cite{Esfahani17} and \cite{gao2016distributionally}. The authors present an algorithm in detail to show how these tractable formulations can be solved using the bisection method. \par
In an earlier paper, \cite{pinar2005robust}, the authors introduced the notion of robust profit opportunities of degree $\theta$ which represent investment strategies that still return profit after $\theta=2$ or $\theta=3$ standard deviations in adverse price movement for the underlying risky securities. We have extended this notion to incorporate the concept of distributional ambiguity to conduct our investigation of distributional RPOs. In some sense our work is an integration and advancement of the concepts developed in the previous two works, namely those of \cite{pinar2005robust} and \cite{zhao2017distributionally}. This concludes our overview of the academic literature on notions of robust statistical arbitrage.

\subsection{Notation and Definitions}
This section lays out the notation and definitions used to develop our framework to investigate distributionally robust profit opportunities (DRPOs). The approach taken here is to start with the definition of an RPO and introduce a notion of distributional uncertainty via the Wasserstein distance metric. As such, we include definitions for these terms as well as some commentary on the problem of moments duality result used to formulate the dual problem for DRPOs.

\begin{remark}
The units for portfolio weight vector $w$ are number of shares of each security. The units for (random) security vector $S_1$ are the period 1 end values for security prices. 
\end{remark}

\subsubsection{Robust Profit Opportunities}
The sets of admissible risky portfolio weights for the weak and strong RPO conditions are 
\begin{align*}
 \Gamma^r_w &:= \{ w \in \mathbb{R}^n : w \cdot S_0 \leq 0; \:  w \neq 0 \}, \\
\Gamma^r_s &:= \{ w \in \mathbb{R}^n : w \cdot S_0 < 0 \},
\end{align*}
where $w \cdot S_0$ denotes $w^\top S_0$.
The RPO condition to be evaluated for covariance matrix $\Sigma$ for random vector $S_1$ under probability measure $Q$ for risky portfolios is
$w \cdot \bar{S}_1 - \theta \sqrt{w^\top \Sigma w} \\ \geq 0$.
where $\bar{S}_1 = \mathbb{E}^Q[ S_1 ]$ and $\theta$ denotes the degree of robustness (or level of risk-aversion). Note that portfolio weight vectors $w$ satisfy the homogeneity property (of degree zero) since $w \cdot \bar{S}_1 - \theta \sqrt{w^\top \Sigma w} \geq 0 \implies w_c \cdot \bar{S}_1 - \theta \sqrt{w_c^\top \Sigma w_c} \geq 0$ for $w_c = c w \; \text{and} \; c > 0$. It is the proportions of the holdings in the assets that distinguish $w$ vectors, not their absolute sizes.

For a given measure $Q$ and $\Gamma_ {s}$, no strong RPO (of level $\theta$) means that 
$\sup_{w \in \Gamma_{s}} w \cdot \bar{S}_1 - \theta \sqrt{w^\top \Sigma w} < 0$.
The empirical measure, $Q_N$, is defined as
$Q_N(dz) = \frac{1}{N} \sum_{i=1}^{N} \mathbbm{1}_{s_{(1,i)}} (dz)$.
To simplify the notation, the leading subscript on $s_{(1,i)}$ is suppressed and going forward we refer to the realization of time 1 asset value vector $s_{(1,i)}$ as just $s_i$. In the context of this work, the uncertainty set for probability measures is
$U_{\delta}(Q_N) \\ = \{Q: D_c(Q,Q_N) \leq \delta\}$
where $D_c$ is the optimal transport cost or Wasserstein discrepancy for cost function $c(\,)$ \cite {blanchetMV}.
The definition for $D_c$ is
\small
\begin{equation*}
D_c(Q,Q') = \inf \{ \mathbb{E}^\pi[c(A,B)]: \pi \in \mathcal{P}(\mathbb{R}^n \times \mathbb{R}^n), \pi_A = Q, \pi_B = Q' \} 
\end{equation*}
\normalsize
where $\mathcal{P}$ denotes the space of Borel probability measures and $\pi_A$ and $\pi_B$ denote the distributions of $A$ and $B$. 
Here $A$ denotes $S_A \in \mathbb{R}^n$ and $B$ denotes $S_B \in \mathbb{R}^n$ respectively.
This work uses the cost function $c$ where
$c(u,v) = \| u-v \|^2_2 = \langle u-v, u-v \rangle$.
The sets of admissible risky portfolio weights for the DRPO conditions (given a minimum target portfolio return $\alpha_0$) are 
\begin{align*}
\Gamma^d_w &:= \{ w \in \mathbb{R}^n : w \cdot S_0 \leq 0; \:  w \neq 0; \, \min_{U_{\delta}(Q_N)} \mathbb{E}^{Q}[ w \cdot S_1 ] \geq \alpha_0 \}, \\
\Gamma^d_s &:= \{ w \in \mathbb{R}^n : w \cdot S_0 < 0; \, \min_{U_{\delta}(Q_N)} \mathbb{E}^{Q}[ w \cdot S_1 ] \geq \alpha_0 \}. 
\end{align*}
Using Proposition 1 in \cite{blanchetMV}, these are equivalent to
\begin{align*}
\Gamma^d_w := \{ w \in \mathbb{R}^n : &w \cdot S_0 \leq 0; \:  w \neq 0; \, \mathbb{E}^{Q_N}[ w \cdot S_1 ] \geq \tilde{\alpha}_0 \}, \\
\Gamma^d_s := \{ w \in \mathbb{R}^n : &w \cdot S_0 < 0; \, \mathbb{E}^{Q_N}[ w \cdot S_1 ] \geq \tilde{\alpha}_0 \},
\end{align*}
where $\tilde{\alpha}_0 := \alpha_0 + \sqrt{\delta} \|w\|$. In our version of the problem we use the relaxation $\tilde{\alpha}_0 := \alpha_0$ which amounts to only requiring that the risky portfolio weights $w$ achieve the minimum target portfolio return $\alpha_0$ for the empirical distribution $Q_N$.

\subsubsection{Restatement of Problem of Moments Duality}
In Section 2 we formulate the primal problems for DRPOs for financial markets with risky securities. A key step in our approach is to use duality results to formulate the simpler yet equivalent dual problems. In this context, to enforce the moment constraint $\mathbb{E}^Q[ w \cdot S_1 ] = \alpha$ for $Q \in U_{\delta}(Q_N)$, we appeal to the strong duality of linear semi-infinite programs. The dual problem appears to be more tractable than the primal problem since it only involves the (finite dimensional) reference probability measure $Q_N$ as opposed to a continuum of probability measures. This allows us to solve a nested optimization problem under an empirical measure defined by the chosen data set. A brief restatement of this duality result follows next. See Appendix B of \cite{blanchet2019robust} and Proposition 2 of \cite{blanchetMV} for further details. \par
\textbf{The problem of moments}. Let $X$ be random vector in probability space $(\Omega, \mathcal{B},\mathcal{P})$ and $(\Omega, \mathcal{B},\mathcal{M}^+)$ where $\mathcal{P}$ and $\mathcal{M}^+$ denote the set of measures and non-negative measures respectively, such that Borel measurable functionals $\phi, f_1,\dots,f_k$ are integrable. Let $f = (f_1,\dots,f_k) : \Omega \rightarrow \mathbb{R}^k$ be a vector of moment functionals. For a real valued vector $q \in \mathbb{R}^k$, we are interested in the worst case bound
\begin{equation*}
v(q) := \sup \big( \: \mathbb{E}_\mu [\phi(X)] \: : \: \mathbb{E}_\mu [f(X)] = q \: ; \:\: \mu \in \mathcal{P} \: \big).
\end{equation*}
Adding a constant term by setting $f_0 = \mathbbm{1}_\Omega$, the constraint $\mathbb{E}_\mu[f_0(X)] = 1$, and defining $\tilde{f} = (f_0,f_1,\dots,f_k)$ and $\tilde{q} = (1,q_1,\dots,q_k)$ gives the following reformulation:
\begin{equation*}
v(q) := \sup \big( \: \int \phi(x) d\mu(x) : \int \tilde{f}(x) d\mu(x) = \tilde{q} \: ; \:\: \mu \in \mathcal{M}^+ \: \big).
\end{equation*}
If a certain Slater condition is satistifed, one has the equivalent dual representation for the above:
\begin{prop*}
Let $\mathcal{Q}_{\tilde{f}} = \{ \int \tilde{f}(x) d\mu(x) : \mu \in \mathcal{M}^+ \}$. If $\tilde{q}$ is an interior point of $\mathcal{Q}_{\tilde{f}}$ then 
\[
v(q) = \inf \big( \: \sum_{i=0}^k a_i q_i : \:\: a_i \in \mathbb{R} ; \:\:\: \sum_{i=0}^k a_i \tilde{f}_i(x) \geq \phi(x) \:\: \forall x \in \Omega \: \big).
\]
\end{prop*}

The primal problem is concerned with the worst case expected loss for some objective function $\phi$, under moment constraints.
Note that the primal problem is an infinite dimensional stochastic optimization problem and thus difficult to solve directly. The simplicity and tractability of the dual problem make it quite attractive as an analytical and/or computational tool in our toolkit. \par
The above duality result has been applied by Blanchet et.\,al and many other authors on topics in data driven distributionally robust stochastic optimization such as robust machine learning, portfolio selection, and risk management. For these types of robust optimization problems, the incorporation of distributional uncertainty can be viewed as adding a penalty term (similar to penalized regression) to the optimal solution \cite{blanchetMV}. This gives us a nice intuitive way to think about the \textit{cost} of robustness.

\section{Theory: DRPOs}
This section develops the theory for DRPOs in financial markets with (only) risky assets. Extending the framework to handle markets with risk free assets is quite tractable; however, it has been omitted due to space constraints. Let us focus on the strong conditions (the weak conditions are similar). Both worst case and best case DRPO conditions are developed. Section 2.1 deals with the worst case conditions, meaning that DRPOs of \textit{at least} level $\theta^{wc}$ exist. The primal problem is formulated using the notions discussed in Section 1.3.1. The dual problem is formulated using the problem of moments duality result from Section 1.3.2. Note that the dual problem is a nested stochastic optimization problem. The inner problem (evaluating $\Psi_{\lambda,w}$) and middle problem (evaluating the dual objective function over $\inf_{\lambda_1 \geq 0, \lambda_2}$) can be solved jointly using the techniques from Proposition 3 in \cite{blanchetMV}. Finally, the outer optimization problem (evaluating over $\inf_{w \in \Gamma^d_{s}}$), for the strong case, can be formulated as a finite dimensional convex optimization problem. A similar approach is taken in Section 2.2 for the best case conditions, meaning that DRPOs of \textit{at most} level $\theta^{bc}$ exist. This machinery gives us a practical approach to explore applications of our DRPO framework. Section 2.3 shows how to incorporate portfolio restrictions (such as short sales) in a straightforward manner. \par

\begin{remark}
For our problem setting, the covariance matrix $\Sigma$ is assumed to be positive definite under the reference probability measure $Q_N$. Furthermore, the portfolio is assumed to consist of $n \geq 1$ risky securities (excluding the risk-free security), with short sales allowed. 
\end{remark}

\subsection{Worst Case DRPO Conditions}
We extend the approach in \cite{pinar2005robust} to arrive at these DRPO conditions. The authors define an RPO of degree $\theta$ as a portfolio $w \in \mathbb{R}^n$ that satisfies $w \cdot \bar{S}_1 - \theta \sqrt{ w^\top \Sigma w } \geq 0$ \: and \: $w \cdot S_0 < 0$. The authors comment that RPO is related to the notions of \textit{risk-adjusted return} and \textit{Sharpe ratio}. The first condition is equivalent to $\frac{w \cdot \bar{S}_1}{\sqrt{ w^\top \Sigma w }} \geq \theta$. Adding the normalization constraint $w \cdot \bar{S}_1 = \alpha$ for $\alpha > 0$ and simplifying gives $w^\top \Sigma w \leq g_\alpha(\theta)$ where $g_\alpha(\theta) := \alpha^2/\theta^2$. Furthermore, the normalization $w \cdot \bar{S}_1 = \alpha \implies \mathbb{E}^Q[ w \cdot S_1 ] = \alpha$.
For minimum target portfolio return $\alpha_0$, the strong \textit{worst case} DRPO condition can be expressed as
\small
\begin{equation*}\label{strongPwc}
\inf_{ w \in \Gamma^d_s } \: \max_{\alpha \geq \alpha_0} \: { \sup_{ Q \in U_{\delta}(Q_N); \: \mathbb{E}^Q[ w \cdot S_1 ] = \alpha } (w^\top \Sigma w) } \: \leq \: g_\alpha(\theta^{wc}). \:\:  \tag{P\textsuperscript{wc}}
\end{equation*}
\normalsize
Using Proposition 2 in \cite{blanchetMV} which invokes problem of moments duality (see Section 1.3.2), the dual formulation for the inner optimization problem 
\begin{equation*}\label{Iwc}
 \sup_{ Q \in U_{\delta}(Q_N); \:\: \mathbb{E}^Q[ w \cdot S_1 ] = \alpha } (w^\top \mathbb{E}^Q[ S_1 S_1^\top ] w) \tag{I\textsuperscript{wc}}
\end{equation*}
where $w^\top \Sigma w = (w^\top \mathbb{E}^Q[ S_1 S_1^\top ] w) - \alpha^2$ is 
\begin{equation*}\label{strongDwc}
\inf_{ \lambda_1 \geq 0, \lambda_2 } \: [ \: \lambda_1 \delta + \lambda_2 \alpha + \frac{1}{N} \sum_{i=1}^N \Psi^{wc}_{\lambda,w} (s_i) \: ] \tag{D\textsuperscript{wc}}
\end{equation*}
where $\Psi^{wc}_{\lambda,w}$ is defined, in terms of cost function $c(\,)$, as
$\Psi^{wc}_{\lambda,w}(s_i) \\ = \sup_{\tilde{s} \in \mathbb{R}^n} [ \, ( w \cdot \tilde{s} )^2 - \lambda_1 c( \tilde{s}, s_i) - \lambda_2 (w \cdot \tilde{s}) \, ]$.

\subsubsection{Inner and Middle Optimization Problems}
The goal here is to evaluate 
\begin{equation} 
\inf_{ \lambda_1 \geq 0, \lambda_2 } \: [ \: \lambda_1 \delta + \lambda_2 \alpha + \frac{1}{N} \sum_{i=1}^N \Psi^{wc}_{\lambda,w} (s_i) \: ] 
\end{equation}
in closed form. Using Proposition 3 and Theorem 1 in \cite{blanchetMV} it follows that when
\small
\[
\delta \| w \|^2_2 - ( \alpha - \mathbb{E}^{Q_N}[w \cdot S_1] )^2 \geq 0 \implies \sup_{ Q \in U_{\delta}(Q_N); \\ \mathbb{E}^Q[ w \cdot S_1 ] = \alpha } (w^\top \mathbb{E}^Q[ S_1 S_1^\top ] w)
\] 
\normalsize
\text{is feasible,}
then for $w^\top \Sigma w = (w^\top \mathbb{E}^Q[ S_1 S_1^\top ] w) - \alpha^2$
\[
\max_{\alpha \geq \alpha_0; \: \delta \| w \|^2_2 - ( \alpha - \mathbb{E}^{Q_N}[w \cdot S_1] )^2 \geq 0} \:\: \sup_{ Q \in U_{\delta}(Q_N); \:\: \mathbb{E}^Q[ w \cdot S_1 ] = \alpha } (w^\top \Sigma w)
\]
evaluates to
\begin{equation}
\big(\sqrt{w^\top \Sigma w} + \sqrt{\delta} \|w\|_2 \big)^2
\end{equation}
where $\Sigma$ is evaluated under the reference measure $Q_N$ and the optimal $\alpha^* := \mathbb{E}^{Q_N}[ w \cdot S_1 ] \geq \alpha_0$. 

\subsubsection{Outer Optimization Problem}
The strong \textit{worst case} DRPO condition (\ref{strongPwc}) is \textit{now} 
\small
\begin{equation*}\label{s2wc}
v^{wc}_{\alpha^*}(\delta) := \inf_{ w \in \Gamma^d_s } \: \big(\sqrt{w^\top \Sigma w} + \sqrt{\delta} \|w\|_2 \big)^2 \leq g_{\alpha^*}( \theta^{wc}). \tag{D2\textsuperscript{wc}}
\end{equation*}
\normalsize
\begin{theorem}
$v^{wc}_{\alpha^*}(\delta)$ can be computed by solving convex nonlinear program (NLP) N\_SRPO\textsuperscript{wc} (listed below).
\end{theorem}
\noindent Note this is essentially a second order conic program (SOCP).
\begin{mini}|l|
{\substack{w \in \mathbb{R}^n}}
{ \sqrt{w^\top \Sigma w} + \sqrt{\delta} \|w\|_2 }{\label{N_SRPO}}
{}
\addConstraint{w \cdot S_0}{\leq - \epsilon}
\addConstraint{\frac{1}{N} \sum_{i=1}^N w \cdot s_i}{\geq \tilde{\alpha}_0}.
\end{mini}

\begin{proof}
The formulation is straightforward. The constraint set $w \in \Gamma^d_s$ is readily obtained via the constraint $w \cdot S_0 \leq - \epsilon$ for a suitably small choice of $\epsilon > 0$. The first moment constraint $\mathbb{E}^{Q_N}[ w \cdot S_1 ] \geq \tilde{\alpha}_0$ is described as above. The squaring in the original objective function does not change the optimal solution. It follows that N\_SRPO\textsuperscript{wc} is a convex SOCP, solvable via standard solvers.
\end{proof}


\begin{theorem}
For a given $\theta^{wc}$, the critical radius $\delta^{wc}_{\alpha^*}$ can be expressed as $\inf \{\delta \geq 0 :v^{wc}_{\alpha^*}(\delta) \geq g_{\alpha^*}(\theta^{wc})\}$. Furthermore, $\delta^{wc}_{\alpha^*}$ can be explicitly computed via binary search. Let $\delta_{\alpha^*} < \delta^{wc}_{\alpha^*}$. For $Q \in U_{\delta_{\alpha^*}}(Q_N)$, it follows that $Q$ admits strong RPOs of \textit{at least} level $\theta^{wc}$. For $Q \notin U_{\delta^{wc}_{\alpha^*}}(Q_N)$, it follows that $Q$ \textit{may} admit strong RPOs of levels less than $\theta^{wc}$.
\end{theorem}
\begin{proof}
This characterization of the critical radius $\delta^{wc}_{\alpha^*}$ follows from the condition (\ref{s2wc}) as well as the definition of DRPOs (see Section 1.3.1). The asymptotic properties of $v^{wc}_{\alpha^*}$ are such that $v^{wc}_{\alpha^*}(0) \geq 0$ and $\lim_{\delta \to \infty} v^{wc}_{\alpha^*}(\delta) \geq g_{\alpha^*}(\theta^{wc})$. Furthermore, since $v^{wc}_{\alpha^*}(\delta)$ is a non-decreasing function of $\delta$, it follows that $\delta^{wc}_{\alpha^*}$ can be computed via binary search.
\end{proof}

\begin{remark}
One can view the critical radius $\delta^{wc}_{\alpha^*}$ as a relative measure of the \textit{degree} of strong RPO in the reference measure $Q_N$. Those $Q_N$ which are ``close" to admitting RPOs of level less than $\theta^{wc}$ will have a relatively smaller value of $\delta^{wc}_{\alpha^*}$.
\end{remark}

\subsection{Best Case DRPO Conditions}
We follow the approach from the previous subsection. To reflect the base case outcome (inside the Wasserstein ball of probability measures of radius $\delta$), replace the $\sup$ with $\inf$ and $\max$ with $\min$.
The strong \textit{best case} DRPO condition is
\small
\begin{equation*}\label{Pbc}
\inf_{ w \in \Gamma^d_s } \: \min_{\alpha \geq \alpha_0} \: { \inf_{ Q \in U_{\delta}(Q_N); \:\: \mathbb{E}^Q[ w \cdot S_1 ] = \alpha } (w^\top \Sigma w) } \:\: \geq \:\: g_\alpha(\theta^{bc}) \tag{P\textsuperscript{bc}}
\end{equation*}
\normalsize
Using Proposition 2 in \cite{blanchetMV} which invokes problem of moments duality (see Section 1.3.2), the dual formulation for the inner optimization problem
\begin{equation*}\label{Ibc}
 \sup_{ Q \in U_{\delta}(Q_N); \:\: \mathbb{E}^Q[ w \cdot S_1 ] = \alpha } - (w^\top \mathbb{E}^Q[ S_1 S_1^\top ] w) \tag{I\textsuperscript{bc}}
\end{equation*}
where $w^\top \Sigma w = (w^\top \mathbb{E}^Q[ S_1 S_1^\top ] w) - \alpha^2$ is 
\begin{equation*}\label{Dbc}
\inf_{ \lambda_1 \geq 0, \lambda_2 } \: [ \: \lambda_1 \delta + \lambda_2 \alpha + \frac{1}{N} \sum_{i=1}^N \Psi^{bc}_{\lambda,w} (s_i) \: ] \tag{D\textsuperscript{bc}}
\end{equation*}
where $\Psi^{bc}_{\lambda,w}$ is defined, in terms of cost function $c(\,)$, as \\
$\Psi^{bc}_{\lambda,w}(s_i) = \sup_{\tilde{s} \in \mathbb{R}^n} [ \, -( w \cdot \tilde{s} )^2 - \lambda_1 c( \tilde{s}, s_i) - \lambda_2 (w \cdot \tilde{s}) \, ]$.

\subsubsection{Inner and Middle Optimization Problems}
The goal here is to evaluate 
\[ 
- \left\{ \inf_{ \lambda_1 \geq 0, \lambda_2 } \: \bigg[ \: \lambda_1 \delta + \lambda_2 \alpha + \frac{1}{N} \sum_{i=1}^N \Psi^{bc}_{\lambda,w} (s_i) \: \bigg] \right\} 
\]
in closed form. 
\begin{prop}
Using techniques from Proposition 3 and Theorem 1 in \cite{blanchetMV} it follows that when
\begin{align*}
\delta \| w \|^2_2 - ( \alpha - \mathbb{E}^{Q_N}[w \cdot S_1] )^2 \geq 0 \implies \\
\left\{ \sup_{ Q \in U_{\delta}(Q_N); \:\: \mathbb{E}^Q[ w \cdot S_1 ] = \alpha } - (w^\top \mathbb{E}^Q[ S_1 S_1^\top ] w) \right\}
\end{align*}
\text{is feasible,}
then for $w^\top \Sigma w = (w^\top \mathbb{E}^Q[ S_1 S_1^\top ] w) - \alpha^2$
\begin{equation}
\min_{\alpha \geq \alpha_0; \: \delta \| w \|^2_2 - ( \alpha - \mathbb{E}^{Q_N}[w \cdot S_1] )^2 \geq 0} \:\: \inf_{ Q \in U_{\delta}(Q_N); \:\: \mathbb{E}^Q[ w \cdot S_1 ] = \alpha } (w^\top \Sigma w)
\end{equation}
evaluates to
\begin{equation}
\max{\big(\sqrt{w^\top \Sigma w} - \sqrt{\delta} \|w\|_2, 0 \big)}^2
\end{equation}
where $\Sigma$ is evaluated under the reference measure $Q_N$ and the optimal $\alpha^* := \mathbb{E}^{Q_N}[ w \cdot S_1 ] \geq \alpha_0$.  
\end{prop}
\begin{proof}
The proof consists of a series of steps. First one determines that $\Psi^{bc}_{\lambda,w}$ is well defined due to the (leading) negative quadratic term for $\Psi^{bc}_{\lambda,w}$. Next one evaluates first order optimality conditions for the dual formulation with respect to $\lambda_1 \geq 0$ and $\lambda_2$. The feasibility condition $\delta \| w \|^2_2 - ( \alpha - \mathbb{E}^{Q_N}[w \cdot S_1] )^2 \geq 0$ arises when evaluating optimality with repsect to $\lambda_1$. Then, using back-substitution and simplifying one arrives at the functional form in (4). Note that portfolio variance is non-negative (always) hence the zero floor induced by the $\max$ operator is sensible. See the Appendix for the detailed proof.
\end{proof}
\begin{remark}
It is interesting to note that the worst case and best case portfolio variances are symmetric with penalty and benefit terms $\sqrt{\delta} \| w \|_2$ respectively.
However, since variance is inherently a non-negative quantity, the best case portfolio variance is floored at zero. Furthermore, zero variance may lead to a classical arbitrage situation. Indeed, this is the case if $\:\: \exists \:\: w \in \Gamma^r_{s} \:\: \text{such that} \:\: w^\top \Sigma w = 0 \:\: \wedge \:\: w \cdot \tilde{S}_1 \geq 0$ where $\Sigma$ is evaluated under some probability measure $Q \in U_{\delta}(Q_N)$ \cite{pinar2005robust}.
\end{remark}
\subsubsection{Outer Optimization Problem}
The strong \textit{best case} DRPO condition (\ref{Pbc}) is \textit{now} 
\small
\begin{equation*}\label{D2bc}
v^{bc}_{\alpha^*}(\delta) := \inf_{ w \in \Gamma^d_s } \: \max{\big(\sqrt{w^\top \Sigma w} - \sqrt{\delta} \|w\|_2, 0 \big)}^2 \geq g_{\alpha^*}( \theta^{bc} ). \tag{D2\textsuperscript{bc}}
\end{equation*}
\normalsize

\begin{theorem}
$v^{bc}_{\alpha^*}(\delta)$ can be computed by solving non-convex nonlinear program (NLP) N\_SRPO\textsuperscript{bc} (listed below).
\end{theorem}
\begin{mini}|l|
{\substack{w \in \mathbb{R}^n}}
{ \max{\big(\sqrt{w^\top \Sigma w} - \sqrt{\delta} \|w\|_2 , 0 \big)}^2 }{\label{N_SRPO}}
{}
\addConstraint{w \cdot S_0}{\leq - \epsilon}
\addConstraint{\frac{1}{N} \sum_{i=1}^N w \cdot s_i}{\geq \tilde{\alpha}_0}.
\end{mini}

\begin{proof}
Again, the formulation is straightforward. The constraint set $w \in \Gamma^r_s$ is readily obtained via the constraint $w \cdot S_0 \leq - \epsilon$ for a suitably small choice of $\epsilon > 0$. The first moment constraint $\mathbb{E}^{Q_N}[ w \cdot S_1 ] \geq \tilde{\alpha}_0$ is described as above. Note that the mapping $w \rightarrow w^\top \Sigma w$ is convex but the objective function is non-convex. It follows that N\_SRPO\textsuperscript{bc} is a non-convex nonlinear program solvable via standard solvers.
\end{proof}

\begin{corollary}
$w^* \in \argmin_{w \in \mathbb{R}^n} \sqrt{w^\top \Sigma w} - \sqrt{\delta} \|w\|_2 \\ \implies w^* \in \argmin_{w \in \mathbb{R}^n} \max{\big(\sqrt{w^\top \Sigma w} - \sqrt{\delta} \|w\|_2, 0 \big)}^2$ therefore $v^{bc}_{\alpha^*}(\delta)$ can be computed by solving non-convex nonlinear program (NLP) N\_SRPO2\textsuperscript{bc} (listed below).
\end{corollary}
\begin{mini}|l|
{\substack{w \in \mathbb{R}^n}}
{ \sqrt{w^\top \Sigma w} - \sqrt{\delta} \|w\|_2 }{\label{N_SRPO2}}
{}
\addConstraint{w \cdot S_0}{\leq - \epsilon}
\addConstraint{\frac{1}{N} \sum_{i=1}^N w \cdot s_i}{\geq \tilde{\alpha}_0}.
\end{mini}

\begin{proof}
This follows by observing that $\max (g(w),0)^2$ is a monotonic (non-decreasing) transformation of $g(w)$.
\end{proof}

\begin{prop}
Solving N\_SRPO2\textsuperscript{bc} is equivalent to solving up to three one-dimensional search problems $\min_{t > 0} \sqrt{f(t)} \\ -\sqrt{\delta t}$ where 
$f(t)$ is the optimal value of a parameterized SDP problem.
\end{prop}
\begin{proof}
The proof uses results about semidefinite programming (SDP) relaxations of quadratic minimization problems. See the Appendix for details.
\end{proof}


\begin{theorem}
For a given $\theta^{bc}$, the critical radius $\delta^{bc}_{\alpha^*}$ can be expressed as $\inf \{\delta \geq 0 :v^{bc}_{\alpha^*}(\delta) \leq g_{\alpha^*}(\theta^{bc})\}$. Furthermore, $\delta^{bc}_{\alpha^*}$ can be explicitly computed via binary search. Let $\delta_{\alpha^*} < \delta^{bc}_{\alpha^*}$. For $Q \in U_{\delta_{\alpha^*}}(Q_N)$, it follows that $Q$ allows strong RPOs of \textit{at most} degree $\theta^{bc}$. For $Q \notin U_{\delta^{bc}_{\alpha^*}}(Q_N)$, it follows that $Q$ \textit{may} allow strong RPOs of \textit{more than} degree $\theta^{bc}$.
\end{theorem}
\begin{proof}
This characterization of the critical radius $\delta^{bc}_{\alpha^*}$ follows from the condition (\ref{D2bc}) as well as the definition of DRPOs (see Section 1.3.1). The asymptotic properties of $v^{bc}_{\alpha^*}$ are such that $v^{bc}_{\alpha^*}(0) \geq 0$ and $\lim_{\delta \to \infty} v^{bc}_{\alpha^*}(\delta) \leq g_{\alpha^*}(\theta^{bc})$. Furthermore, since $v^{bc}_{\alpha^*}(\delta)$ is a non-increasing function of $\delta$, it follows that $\delta^{bc}_{\alpha^*}$ can be computed via binary search.
\end{proof}
\begin{remark}
One can view the critical radius $\delta^{bc}_{\alpha^*}$ as a relative measure of the \textit{degree} of strong RPO in the reference measure $Q_N$. Those $Q_N$ which are ``close" to admitting RPOs of levels more than $\theta^{bc}$ will have a relatively smaller value of $\delta^{bc}_{\alpha^*}$.
\end{remark}

\subsection{Portfolio Restrictions}
This subsection discusses refinements to the DRPO conditions (see Sections 2.1 and 2.2) to characterize portfolio restrictions such as short sales restrictions, min and max position constraints, and cardinality constraints \cite{cornuejols2018optimization}. For efficiency of presentation, we refer the reader to the N\_SRPO NLP problems discussed in Sections 2.1 and 2.2 and do not restate those formulations here. An advantage of the computational machinery developed in this paper is that such portfolio restrictions can be readily incorporated into the existing framework. Table 1 (above) describes the various portfolio restrictions discussed here and associated constraints. Others are possible as well. Note that the index set is $j \in \{1,\dots,n\}$ which is suppressed for brevity.
\renewcommand{\arraystretch}{1.5}
\begin{table}[H]
\normalsize
\begin{center}
\caption{Portfolio Restrictions}
\begin{tabular}{ |c|c|l| }
 \hline
\textit{Restriction} & \textit{MINLP Constraint} & \textit{No Restriction} \\
 \hline
Short Sales & $w_j \geq ss_j$ & $ss_j = -M$ \\ 
\hline
Min Positions & $|w_j| \geq \underline{w}$ & $\underline{w} = 0$ \\ 
 \hline
Max Positions & $|w_j| \leq \overline{w}$ & $\overline{w} = M$ \\ 
\hline
Cardinality & $\sum_{j=1}^n \mathbbm{1}_{\{ |w_j| \geq \epsilon \}} \leq m$ & $m=n$ \\ 
\hline
Allocations & $| \sum_{j \in A_k} w_j |  \leq \overline{A_k}$ & $\overline{A_k}=M n$ \\ 
\hline
\end{tabular} 
\end{center} 
\end{table}
\renewcommand{\arraystretch}{1}


\section{Case Study}
This case study investigates the DRPOs for a five year historical data set (of month end closing prices) from July 2015 to June 2020 for a basket of exchange traded funds (ETFs) spread across different sectors of the economy. The 60 month end closing prices define the empirical distribution for random vector $S_1$ and the most recent closing values define $S_0$. The best and worst case critical values of $\theta$ are computed for a trajectory of Wasserstein radii $\delta$. The Matlab \textit{fmincon} solver is used, along with multiple search paths, to arrive at a more robust solution. The critical values are shown in the tables and graphs. Note that $\theta^* = \infty$ denotes the presence of classical arbitrage. For the worst case trajectory, shown in Figure 1, we see that it takes a relatively large value of $\delta$ to bring $\theta^* < 1$. On the other hand, for the best case trajectory, shown in Figure 2, we see that it takes a relatively small value of $\delta$ to bring $\theta^* \rightarrow \infty$.
Intuitively this means that the empirical distribution $Q_N$ is close (in terms of Wasserstein distance) to admitting classical arbitrage.

\begin{table}[!htb]
\begin{center}
\caption{Basket Constituents}
\begin{tabular}{ |c|c|c|c| }
 \hline
  Ticker & Name & Industry & Net Assets (bn) \\
 \hline
 FENY & Fidelity MSCI Energy & Energy & 0.46 \\
 \hline
 JETS & U.S. Global JETS & Travel & 0.93 \\
 \hline
 VGT & Vanguard Tech & Technology & 33.65 \\
 \hline
 VHT & Vanguard Health Care & Health & 12.64 \\
 \hline
 XLF & Financial SPDR Fund & Finance & 17.84 \\
 \hline
\end{tabular}
\end{center}
\end{table}

\begin{table}[!htb]
\begin{center}
\caption{$v_{\alpha^*}^{wc}(\delta)$: Worst Case degree $\theta^*$}
\begin{tabular}{ |r|r|r|r|r|r|r|r| }
 \hline
$\delta$ & 1 & 10 & 100 & 250 & 500 & 1000 \\
 \hline
$\theta^*$ & 2.45 & 2.03 & 1.37 & 1.08 & 0.87 & 0.68 \\
 \hline
\end{tabular}
\end{center}
\end{table}

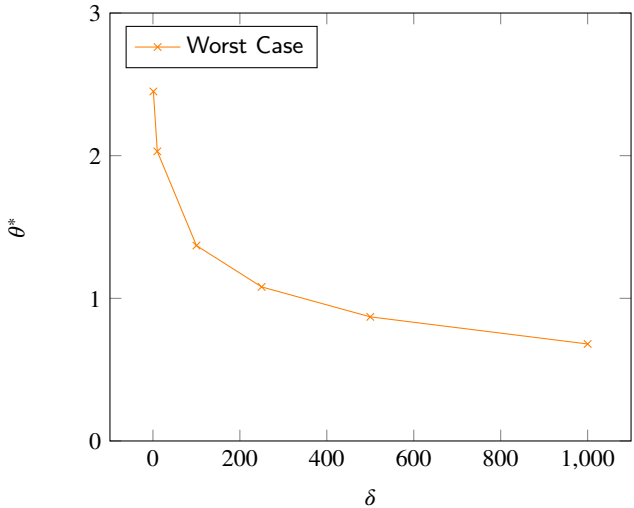
\begin{figure}[!htb]
\caption{Worst Case degree $\theta^*$}
\begin{center}
\begin{tikzpicture}
	\begin{axis}[legend pos=north west,
		xlabel=$\delta$,
		ylabel=$\theta^*$,
		ymin = 0.0, ymax = 3.0 ]
	\addplot[color=orange,mark=x] coordinates {
		(1,2.45)
		(10,2.03)
		(100,1.37)
		(250,1.08)
		(500,0.87)
		(1000,0.68)
	}; \label{plot4_y1}

	\addlegendimage{/pgfplots/refstyle=plot4_y1}\addlegendentry{Worst Case}
	\end{axis}
	
\end{tikzpicture}
\end{center}
\end{figure}

\begin{table}[!htb]
\begin{center}
\caption{$v_{\alpha^*}^{bc}(\delta)$: Best Case degree $\theta^*$}
\begin{tabular}{ |r|r|r|r|r|r|r| }
 \hline
$\delta$ & 0.001 & 0.5 & 1.0 & 2.0 & 5.0 & 10.0 \\
 \hline
$\theta^*$ & 2.83 & 3.92 & 5.78 & $\infty$ & $\infty$ & $\infty$ \\
 \hline
\end{tabular}
\end{center}
\end{table}

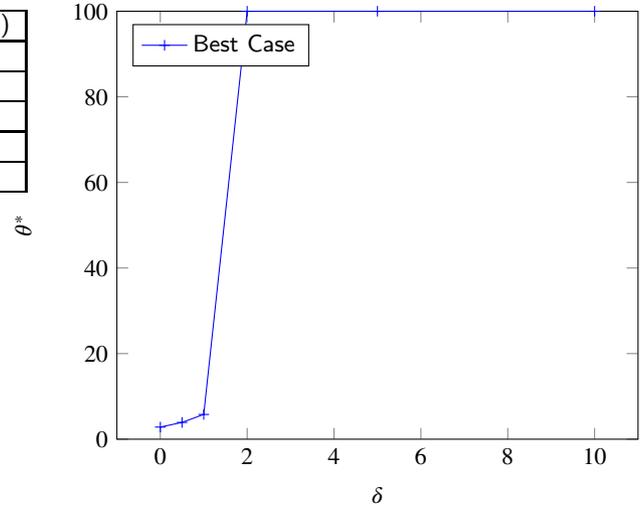
\begin{figure}[!htb]
\caption{Best Case degree $\theta^*$}
\begin{center}
\begin{tikzpicture}
	\begin{axis}[legend pos=north west,
		xlabel=$\delta$,
		ylabel=$\theta^*$,
		ymin = 0.0, ymax = 100 ]
	\addplot[color=blue,mark=+] coordinates {
		(0.001,2.83)
		(0.5,3.92)
		(1.0,5.78)
		(2.0,100)
		(5,100)
		(10,100)
	}; \label{plot5_y1}

	\addlegendimage{/pgfplots/refstyle=plot5_y1}\addlegendentry{Best Case}
	\end{axis}
	
\end{tikzpicture}
\end{center}
\end{figure}

\section{Conclusions and Further Work}
This work has developed theoretical results and investigated calculations of distributionally robust profit opportunities using Wasserstein distance as an ambiguity measure. The financial market overview and foundational notation and problem definitions were introduced in Section 1. Using recent duality results \cite {blanchetFirst}, the simpler dual formulation and its mixture of analytic and computational solutions were derived in Section 2. A case study was investigated in Section 3. Finally, we conclude with some commentary on directions for further research. One direction (as previously mentioned) is to extend the framework to incorporate risk free securities. Another direction is to consider reward-risk ratios other then the Sharpe ratio; a couple such examples would be the Sortino ratio and the CVaR ratio. \par

\section*{Data Availability Statement}
The raw and/or processed data required to reproduce the findings from this research can be obtained from the corresponding author, [D.S.], upon reasonable request.

\section*{Conflict of Interest Statement}
The authors declare they have no conflict of interest.

\section*{Funding Statement}
The authors received no specific funding for this work.

\bibliographystyle{cas-model2-names}

\bibliography{DRPOdcV2}


\appendix
\section{Proof of Proposition 2.1}
\begingroup
\small
\begin{prop*}
Using techniques from Proposition 3 and Theorem 1 in \cite{blanchetMV} it follows that when
\begin{align*}
\delta \| w \|^2_2 - ( \alpha - \mathbb{E}^{Q_N}[w \cdot S_1] )^2 \geq 0 \implies \\
\left\{ \sup_{ Q \in U_{\delta}(Q_N); \:\: \mathbb{E}^Q[ w \cdot S_1 ] = \alpha } - (w^\top \mathbb{E}^Q[ S_1 S_1^\top ] w) \right\}  
\end{align*}
\text{is feasible,} then for $w^\top \Sigma w = (w^\top \mathbb{E}^Q[ S_1 S_1^\top ] w) - \alpha^2$
\begin{equation}
\inf_{ Q \in U_{\delta}(Q_N); \:\: \mathbb{E}^Q[ w \cdot S_1 ] = \alpha } (w^\top \Sigma w)
\end{equation}
evaluates to
\begin{equation}
\max{\big(\sqrt{w^\top \Sigma w} - \sqrt{\delta} \|w\|_2, 0 \big)}^2
\end{equation}
where $\Sigma$ is evaluated under the reference measure $Q_N$ and the optimal $\alpha^* := \mathbb{E}^{Q_N}[ w \cdot S_1 ] \geq \alpha_0$.   
\end{prop*}
\begin{proof}
We apply techniques similar to Proposition 3 from \cite{blanchetMV} and map our notation to align with that paper for convenience of comparison. Towards that end we make the following substitutions: $\{Q,Q_N,\Psi,w,s_i,s\} \rightarrow \{P,P_N,\Phi,\phi,R_i,u\}$ respectively, and translate notation back for the final result. Using the new notation, the dual for problem \ref{Ibc} now becomes
\begin{equation}
- \left\{ \inf_{ \lambda_1 \geq 0, \lambda_2 } \: \bigg[ \: \lambda_1 \delta + \lambda_2 \alpha + \frac{1}{N} \sum_{i=1}^N \Phi (R_i) \: \bigg] \right\} 
\end{equation}
where
\begin{equation}
\Phi(R_i) = \sup_{u \in \mathbb{R}^n} \:\: [ \, -( \phi \cdot u )^2 - \lambda_1 c( u, R_i) - \lambda_2 (\phi \cdot u) \, ].
\end{equation}
Similarly, for $\phi^\top \Sigma \phi = (\phi^\top \mathbb{E}^Q[ R R^\top ] \phi) - \alpha^2$, (5) now becomes 
\begin{equation}
\inf_{ P \in U_{\delta}(P_N); \:\: \mathbb{E}^P[ \phi \cdot R ] = \alpha } (\phi^\top \Sigma \phi).
\end{equation}
Expanding the cost function $c(u,v) = \|u - v\|^2_2$ and making the substitution $\Delta = u - R_i$ gives
\begin{align}
\Phi(R_i) &= \sup_{\Delta} \:\: [ \, -( \phi \cdot (\Delta + R_i ))^2 - \lambda_1 \| \Delta \|^2_2 - \lambda_2 (\phi \cdot (\Delta+ R_i)) \, ] \nonumber \\
			&= \sup_{\Delta} \:\: [ \, -(\phi \cdot \Delta )^2 -2 (\phi \cdot R_i)(\phi \cdot \Delta) - \lambda_1 \| \Delta \|^2_2  - \lambda_2 (\phi \cdot \Delta) \, ] \nonumber \\
             &-( \phi \cdot R_i )^2 -\lambda_2 \phi \cdot R_i \nonumber \\
			&= \sup_{\Delta} \:\: [ \, -(\| \phi \|^2_2+ \lambda_1 ) \| \Delta \|^2_2 +2 |\phi \cdot R_i + \lambda_2 | \| \phi \|_2 \| \Delta \|_2 \, ] \nonumber \\
			&-( \phi \cdot R_i )^2 -\lambda_2 \phi \cdot R_i \nonumber \\
			&= -( \phi \cdot R_i )^2 -\lambda_2 \phi \cdot R_i + \frac{ (2 \phi \cdot R_i + \lambda_2)^2 \| \phi \|^2_2 }{ 4 ( \| \phi \|^2_2 + \lambda_1) }.
\end{align}
Hence $- \left\{ \inf_{ \lambda_1 \geq 0, \lambda_2 } \: \bigg[ \: \lambda_1 \delta + \lambda_2 \alpha + \frac{1}{N} \sum_{i=1}^N \Phi (R_i) \: \bigg] \right\}$ becomes
\begin{align}
&-\inf_{ \lambda_1 \geq 0, \lambda_2 } \: H \: := \frac{1}{N} \sum_{i=1}^N \big[  -( \phi \cdot R_i )^2 -\lambda_2 \phi \cdot R_i \nonumber \\
&+ \frac{ (2 \phi \cdot R_i + \lambda_2)^2 \| \phi \|^2_2 }{ 4 ( \| \phi \|^2_2 + \lambda_1) } \:\: \big] + \lambda_1 \delta + \lambda_2 \alpha.
\end{align}
The first order optimality condition for $\lambda_2$ gives 
\[
\frac{ \partial H}{ \partial \lambda_2} = \alpha + \frac{1}{N} \sum_{i=1}^N \bigg[ -(\phi \cdot R_i) + \frac{2 (2 \phi \cdot R_i + \lambda_2) \| \phi \|^2_2}{ 4 ( \| \phi \|^2_2 + \lambda_1 ) } \bigg] = 0.
\]
Recall $\| \phi \| > 0$ hence we obtain $\lambda_2^* = -2 \alpha -2 C \frac{\lambda_1}{\| \phi \|^2_2}$ where $C := \alpha - \phi \cdot \mathbb{E}^{P_N} (R)$.
Indeed, $\lambda_2^*$ is optimal since the second order condition for $\lambda_2$ gives
\[
\frac{ \partial^2 H }{ \partial \lambda_2^2 } = \frac{ \| \phi \|^2_2 }{ 2( \| \phi \|^2_2 + \lambda_1 ) } > 0.
\]
Substituting $\lambda_2^*$ back into $H$ gives $H =$
\begin{align}
&\frac{1}{N} \sum_{i=1}^N ( \phi \cdot R_i )^2 - \inf_{ \lambda_1 \geq 0 } \: \frac{1}{N} \sum_{i=1}^N \bigg[ \frac{ ( \phi \cdot R_i - \alpha - C \frac{ \lambda_1 }{ \| \phi \|^2_2 } )^2 }{ ( \| \phi \|^2_2 + \lambda_1 ) } \bigg] \nonumber \\
&  + \lambda_1 \delta -2 ( \alpha + C \frac{ \lambda_1 }{ \| \phi \|^2_2 } ) C.
\end{align}
Now let \: $\lambda_1 = \kappa - \| \phi \|^2_2 \geq 0 \implies \kappa \geq \| \phi \|^2_2$ \: to get $H =$
\begin{align}
&\frac{1}{N} \sum_{i=1}^N ( \phi \cdot R_i )^2 + 2 \alpha C - 2 C^2 + \| \phi \|^2_2 \delta \nonumber \\ 
&-\inf_{ \kappa \geq \| \phi \|^2_2 } \: \frac{1}{N} \sum_{i=1}^N \bigg[ \frac{ ( \phi \cdot R_i  - \alpha - \frac{ C \kappa }{ \| \phi \|^2_2 } + C )^2 \| \phi \|^2_2 }{ \kappa } \bigg] \nonumber \\
&-\kappa ( \frac{2 C^2}{ \| \phi \|^2_2 } - \delta ).
\end{align}
Partial substitution for $C = \alpha - \phi \cdot \mathbb{E}^{P_N}(R)$ and noting $\frac{1}{N} \sum_{i=1}^N -2 ( \phi \cdot R_i - \phi \cdot \mathbb{E}^{P_N}(R)) \frac{C \kappa }{\| \phi \|^2_2 } = 0$ gives $H =$
\begin{align}
&\mathbb{E}^{P_N}[(\phi \cdot R)^2] + 2 C (\phi \cdot \mathbb{E}^{P_N}(R)) + \delta \| \phi \|^2_2 \nonumber \\
&- \left\{ \inf_{ \kappa \geq \| \phi \|^2_2 } \: \frac{1}{N} \sum_{i=1}^N \bigg[ \frac{( \phi \cdot R_i - \phi \cdot \mathbb{E}^{P_N}(R) )^2 \| \phi \|^2_2 }{\kappa} \bigg]+ \kappa ( \delta - \frac{C^2 }{\| \phi \|^2_2 } ) \right\}.
\end{align}
If $\delta \| \phi \|^2_2 - C^2 < 0$ the solution is unbounded, which implies 
\begin{align*}
\sup_{ Q \in U_{\delta}(Q_N); \:\: \mathbb{E}^Q[ w \cdot S_1 ] = \alpha } - (w^\top \mathbb{E}^Q[ S_1 S_1^\top ] w)
\end{align*}
\text{is not feasible}. Therefore, impose the feasiblity constraint $\delta \| \phi \|^2_2 - C^2 = \delta \| \phi \|^2_2 - ( \alpha - \phi \cdot \mathbb{E}^{P_N}(R) )^2 \geq 0$. To evaluate the $\inf_{\kappa \geq \| \phi \|^2_2}$ expression, first make the substitution $A_i =  \frac{ ( \phi \cdot R_i - \phi \cdot \mathbb{E}^{P_N}(R) )^2 \| \phi \|^2_2 }{\kappa} $ and $B = (\delta - \frac{C^2}{ \| \phi \|^2_2 } )$ to get
$\inf_{ \kappa \geq \| \phi \|^2_2 } \:\: \frac{1}{N} \sum_{i=1}^N \:\: \big[ \frac{A_i}{\kappa} \big] + B \kappa$. Note this expression is convex hence for the \textit{unconstrained} problem, the first order optimality condition $- \frac{1}{N} \sum_{i=1}^N \frac{A_i}{\kappa^2} + B = 0$ suffices to determine $\kappa^*$.
Some algebra gives $\kappa^* = \sqrt{ \frac{ \frac{1}{N} \sum_{i=1}^N  A_i }{B} } \implies \inf_{ \kappa \geq 0 } \:\: \frac{1}{N} \sum_{i=1}^N \:\: \big[ \frac{A_i}{\kappa} \big] + B \kappa = 2 \sqrt{ \frac{1}{N} \sum_{i=1}^N A_i } \sqrt{B}$. This can be rewritten as
\begin{align}
&\inf_{ \kappa \geq 0 } \:\: \frac{1}{N} \sum_{i=1}^N \:\: \bigg[ \frac{A_i}{\kappa} \bigg] + B \kappa \nonumber \\
&= 2 \sqrt{ \phi^\top [ \frac{1}{N} \sum_{i=1}^N (R_i - \mathbb{E}^{P_N}(R)) (R_i - \mathbb{E}^{P_N}(R))^\top ] \phi } \sqrt{ \delta \| \phi \|^2_2 - C^2 } \nonumber \\
&= 2 \sqrt{ \phi^\top \Sigma \phi } \sqrt{ \delta \| \phi \|^2_2 - C^2 }.
\end{align}
Note that for the \textit{constrained} problem, $\kappa = \| \phi \|^2_2 \implies$ (8) evaluates to $\alpha^2 \implies$ (9) evaluates to 0. Thus we see that (9) becomes
\begin{equation}
\min_{\delta \| \phi \|^2_2 - C^2 \geq 0} \left\{ \begin{array}{lr}
		\Gamma_0 -2 \sqrt{ \phi^\top \Sigma \phi } \sqrt{ \delta \| \phi \|^2_2 - C^2 } - \alpha^2, & \text{for } k^* \geq \| \phi \|^2_2 \\
		0, & \text{otherwise}
		\end{array} \right\}
\end{equation}
where $\Gamma_0 = \mathbb{E}^{P_N}[(\phi \cdot R)^2] + 2 C (\phi \cdot \mathbb{E}^{P_N}(R)) + \delta \| \phi \|^2_2$.
Let us substitute for $C = \alpha - \phi \cdot \mathbb{E}^{P_N}(R)$ and do some work to expand and simplify the long first term inside the $\min$ expression for (14), call it $V_1$, to get
\begin{align*}
V_1 &= \mathbb{E}^{P_N}[(\phi \cdot R)^2]  - ( \phi \cdot \mathbb{E}^{P_N}(R) )^2 + \delta \| \phi \|^2_2 - \alpha^2 + 2 \alpha (\phi \cdot \mathbb{E}^{P_N}(R)) \\
&- (\phi \cdot \mathbb{E}^{P_N}(R))^2 - 2 \sqrt{ \phi^\top \Sigma \phi } \sqrt{ \delta \| \phi \|^2_2 - (\alpha - \phi \cdot \mathbb{E}^{P_N}(R))^2} \\
&= \phi^\top \Sigma \phi + \big[ \delta \| \phi \|^2_2 - (\alpha - \phi \cdot \mathbb{E}^{P_N}(R))^2 \big] \\
&- 2 \sqrt{ \phi^\top \Sigma \phi } \sqrt{ \delta \| \phi \|^2_2 - (\alpha - \phi \cdot \mathbb{E}^{P_N}(R))^2} \\
	 &= \bigg(\sqrt{\phi^\top \Sigma \phi } - \sqrt{ \delta \| \phi \|^2_2 - (\alpha - \phi \cdot \mathbb{E}^{P_N}(R))^2 } \:\: \bigg)^2.
\end{align*}
Now (14) can be written as $\min_{\delta \| \phi \|^2_2 -  (\alpha - \phi \cdot \mathbb{E}^{P_N}(R))^2 \geq 0}$
\begin{equation}
 \left\{ \begin{array}{lr}
		\bigg(\sqrt{\phi^\top \Sigma \phi } - \sqrt{ \delta \| \phi \|^2_2 - (\alpha - \phi \cdot \mathbb{E}^{P_N}(R))^2 } \:\: \bigg)^2, & \text{for } k^* \geq \| \phi \|^2_2 \\
		0, & \text{otherwise}
		\end{array} \right\}.
\end{equation}
Observing that $\alpha = \phi \cdot \mathbb{E}^{P_N}(R)$ realizes the minimum, and $\| \phi \| \neq 0$, it follows that (15) reduces to
\begin{equation}
		\left\{ \begin{array}{lr}
		\bigg(\sqrt{\phi^\top \Sigma \phi } - \sqrt{ \delta } \: \| \phi \|_2 \:\: \bigg)^2, & \text{for } k^* \geq \| \phi \|^2_2 \\
		0, & \text{otherwise}
		\end{array} \right\}.
\end{equation}
Next, we proceed to evaluate the condition $\kappa^* \geq \| \phi \|^2_2$. Recall $\kappa^* = \sqrt{ \frac{ \frac{1}{N} \sum_{i=1}^N  A_i }{B} }$. For $\alpha$ as above, this simplifies to $\kappa^* = \sqrt{ \frac{ \phi^\top \Sigma \phi }{\delta} } \| \phi \|_2$. The condition $\kappa^* \geq \| \phi \|^2_2$ now becomes $\sqrt{ \phi^\top \Sigma \phi } \geq \sqrt{ \delta } \| \phi \|_2$. Therefore (16) simplifies to \\
$\max{\big(\sqrt{\phi^\top \Sigma \phi} - \sqrt{\delta} \|\phi \|_2, 0 \big)}^2 \implies$ 
\begin{equation}
\inf_{ Q \in U_{\delta}(Q_N); \:\: \mathbb{E}^Q[ w \cdot S_1 ] = \alpha } (w^\top \Sigma w) = \max{\big(\sqrt{w^\top \Sigma w} - \sqrt{\delta} \| w \|_2, 0 \big)}^2
\end{equation}
and we are done.
\end{proof}
\endgroup

\section{Proof of Proposition 2.2}
\begingroup
\small
\begin{prop*}
Solving N\_SRPO2\textsuperscript{bc} is equivalent to solving up to three one-dimensional search problems $\min_{t > 0} \sqrt{f(t)} -\sqrt{\delta t}$ where 
$f(t)$ is the optimal value of a parameterized SDP problem.
\end{prop*}
\begin{proof}
Consider the reformulation of N\_SRPO2\textsuperscript{bc} given by
\begin{mini}|l|
{\substack{w \in \mathbb{R}^n}}
{ \big(\sqrt{w^\top \Sigma w} - \sqrt{\delta} \|w\|_2 \big) }{\label{N_SRPO2b}}
{}
\addConstraint{a^\top w}{\geq 1}
\addConstraint{b^\top w}{\geq 1}.
\end{mini}
The KKT optimality condition says that
\begin{equation}
\frac{\Sigma w}{\sqrt{w^\top \Sigma w}} - \frac{\sqrt{\delta} w}{\sqrt{w^\top w}} = \beta_1 a + \beta_2 b
\end{equation}
where $\beta_1 \geq 0$ and $\beta_2 \geq 0$ are the Lagrange multipliers associated with the linear constraints. For the purpose of our discussion (computational efficiency) let us restrict our attention to the case where $\beta_1 > 0, \beta_2 > 0 \implies a^\top w = 1 \wedge b^\top w = 1$. The other cases of either $a^\top w = 1$ or $b^\top w = 1$ can be treated separately. In this case, the two linear constraints eliminate two variables. Let $\tilde{w} \in \mathbb{R}^{n-2}$ denote the remaining $n-2$ variables. Then write the reformulation
\begin{equation}
\min_{\tilde{w} \in \mathbb{R}^{n-2}} \sqrt{q_1(\tilde{w})} - \sqrt{\delta q_2(\tilde{w})}
\end{equation}
where $q_1$ and $q_2$ are non-negative convex quadratic functions. Let $f(t)$ denote the optimal value of
\begin{mini}|l|
{\substack{\tilde{w} \in \mathbb{R}^{n-2}}}
{ q_1(\tilde{w}) }{}
{}
\addConstraint{q_2(\tilde{w})}{= t}.
\end{mini}
By the so-called S-lemma (see \cite{sturm2003cones} and \cite{ye2003new}), the function $f(t)$ is convex and can be evaluated by a parameterized SDP in polynomial time for any given $t$.
Now the reformulation reduces to 
\begin{equation}
\min_{t > 0} \sqrt{f(t)} - \sqrt{\delta t}
\end{equation}
and we are done.
\end{proof}
\endgroup

\end{document}